\PassOptionsToPackage{table}{xcolor}
\documentclass[pdflatex,sn-mathphys-ay]{sn-jnl}

\usepackage{graphicx}
\usepackage{amsmath,amssymb,amsfonts}
\usepackage{array}
\usepackage{booktabs}
\usepackage{float}
\usepackage[title]{appendix}
\usepackage{enumitem}
\setlist{nosep,leftmargin=1.4em}
\graphicspath{{figures/}}

\makeatletter
\AtBeginDocument{\immediate\write\@auxout{\string\citation{snsettings}}}
\makeatother

\newcommand{\pkg}[1]{\textsf{#1}}
\newcommand{\HL}{\mathrm{HL}}
\newcommand{\EF}{\mathrm{EF}}
\newcommand{\bm}[1]{\boldsymbol{#1}}
\newcommand{\pib}{\bar\pi}
\usepackage{tikz}
\definecolor{orcidgreen}{HTML}{A6CE39}
\gdef\orcidlogo{\raisebox{-0.2ex}{\tikz[baseline=-0.6ex]{\fill[orcidgreen] (0,0) circle (0.75ex);
  \node[white, inner sep=0pt, font=\sffamily\bfseries\fontsize{3.6}{3.6}\selectfont] at (0,0) {iD};}}}

\theoremstyle{thmstyleone}
\newtheorem{theorem}{Theorem}
\newtheorem{lemma}{Lemma}
\newtheorem{proposition}{Proposition}
\newtheorem{corollary}{Corollary}
\theoremstyle{thmstyletwo}
\newtheorem{remark}{Remark}

\begin{document}

\title[A modified Hosmer--Lemeshow test for asymmetric links]{A modified Hosmer--Lemeshow goodness-of-fit test for asymmetric links: second-order power and robustness}

\author*[1]{\fnm{Ebrahim Khaled} \sur{Ebrahim}\,\orcid{https://orcid.org/0009-0006-7839-8778}}\email{ebrahimkhaled@alexu.edu.eg}
\author[1]{\fnm{Ibrahim Galal} \sur{Khattab}\,\orcid{https://orcid.org/0009-0006-4690-9828}}\email{ibrahimkhattab213@yahoo.com}
\author[1]{\fnm{Ahmed} \sur{El-Kotory}\,\orcid{https://orcid.org/0009-0000-1485-0768}}\email{ahmed.elkatory@alexu.edu.eg}
\affil[1]{\orgdiv{Department of Applied Statistics}, \orgname{Faculty of Business, Alexandria University}, \orgaddress{\city{Alexandria}, \country{Egypt}}}

\abstract{The Hosmer--Lemeshow test is the routine check of fit for logistic regression, yet it is weak against a wrong link, and the directed tests that are stronger lose their level when a single covariate is recorded in error. We apply Farrington's correction to its risk groups, which removes the single-record terms of the statistic and keeps the within-group products of residuals. The result is the score test for records of similar fitted risk sharing a miscalibration: locally most powerful against a smooth misfit in either direction, and as robust to gross errors as the Hosmer--Lemeshow test. One alignment functional, mapped in advance for asymmetric links, decides when it is more powerful than the Hosmer--Lemeshow test, at second order with ten groups and at first order with many. For the many groups of large samples, normal limits for both statistics correct the chi-squared reference of the widely used grouping rule of Paul and colleagues, under which the Hosmer--Lemeshow test reached a size of 12\% or lost all its power, and they explain why that rule hid a clear misfit in 8873 SUPPORT patients that groups of 25 records detected. In pre-registered simulations, one gross covariate error in a thousand records raised the false-alarm rates of Stukel's, the cubic calibration and the GiViTI tests to 25--39\%, while the risk-grouped tests stayed near 5\%; among these, the corrected test was the most powerful in all 34 settings with negative alignment and a visible difference. The test is in the R package ebrahim.gof.}



\maketitle

\section{Introduction}\label{sec:intro}

The Hosmer--Lemeshow test \citep{Hosmer1980} is the routine check of fit for a logistic regression model: the records are sorted by fitted risk, cut into $G$ groups, usually ten, and observed and expected events are compared by a Pearson-type sum. It has little power against a wrong link \citep{Hosmer1997,Liu2024,Canary2016}. Directed tests do better on clean data: Stukel's score test \citep{Stukel1988}, the cubic calibration test that extends Pregibon's goodness-of-link test \citep{Pregibon1980,Nattino2014} and the GiViTI calibration test \citep{Nattino2014} use each record's residual in a chosen direction, and the tests of \citet{Tsiatis1980} and \citet{Xie2008} group the records by clusters in covariate space.

Real data are rarely clean. The level of a test of fit should not be decided by a few gross recording errors; it should stay controlled in a small neighbourhood of the model \citep{Hampel1974,HeritierRonchetti1994,CantoniRonchetti2001}. A covariate recorded on the wrong scale creates such an error, a high-leverage record whose response does not match its covariate. In our pre-registered screens (Section~\ref{sec:screen}) one such record in a thousand raised the false-alarm rate of Stukel's test from 5\% to 37\%, and ten raised those of the Tsiatis and Xie tests above 55\%, while the tests that group by fitted risk, the Hosmer--Lemeshow test and its variance-corrected version by \citet{PigeonHeyse1999}, stayed close to 5\%. An extreme record forms a cluster of its own in covariate space and receives a large weight in a record-level score; grouped by fitted risk, it is one of many records of similar risk.

This leaves a protected test that is weak against a wrong link. For sparse binomial data, \citet{Farrington1996} added to Pearson's statistic a first-order term that makes it locally orthogonal to the coefficients and minimises its variance. For binary records the corrected statistic equals the sample size (Lemma~\ref{lem:degenerate}), the deficiency that \citet{Kuss2002} found by simulation. Applied to the Hosmer--Lemeshow risk groups, it gives
\[
\EF=\HL-C,\qquad C=\sum_{g=1}^{G}\frac{(1-2\pib_g)(o_g-e_g)}{n_g\pib_g(1-\pib_g)},
\]
with $o_g$, $e_g$, $n_g$ and $\pib_g$ the observed and expected events, the size and the mean fitted risk of group $g$. The statistic was introduced in a preprint by the authors \citep{EbrahimEF}, drawing on the first author's thesis \citep{EbrahimThesis}; the decomposition, the theory, the robustness study and the comparisons here are new.

Our contributions are the following.
\begin{enumerate}
\item \emph{Structure and optimality.} $\EF$ is the Hosmer--Lemeshow statistic with its single-record terms removed, $\HL=G+C+Q$ and $\EF=G+Q$ (Lemma~\ref{lem:decomp}), and it is the $C(\alpha)$ score test for a common correlation of records within risk groups: locally most powerful against a shared miscalibration of either sign (Proposition~\ref{prop:score}). The Pigeon--Heyse test is $\HL$ with a within-group variance and one degree of freedom too many (Proposition~\ref{prop:ph}).
\item \emph{Many groups.} When $G$ grows with $n$, $\EF$ and $\HL$ are asymptotically normal, $\EF$ has minimum variance among linearly corrected versions of $\HL$, and it is more powerful exactly when $A$ is below an explicit threshold (Theorems~\ref{thm:sparse-null} and \ref{thm:sparse-power}). The normal references correct the $\chi^2_{G-2}$ reference under the grouping rule of \citet{Paul2013} recommended by \citet{Hosmer2013}, and the theory explains why that rule loses power against smooth misfits (Sections~\ref{sec:large-sim} and \ref{sec:support}).
\item \emph{Fixed number of groups.} $\EF$ shares the null limit of $\HL$ (Theorem~\ref{thm:null}; for $\HL$ this is the limit of \citealt{Hosmer1980}, made rigorous for data-dependent cells by \citealt{Surjanovic2023}); their power differs at second order by $-\rho f(q)A$, with $A$ an alignment functional of the departure (Proposition~\ref{thm:gauss}, Theorem~\ref{thm:second}). Second-order comparisons are classical for power-divergence statistics \citep{Drost1989,ReadCressie1988,Taneichi2002,SekiyaTaneichi2004}; for corrected Hosmer--Lemeshow statistics we know of none.
\item \emph{Robustness.} One corrupted record moves $\EF$ and $\HL$ by a bounded amount whatever its covariates, and Stukel's statistic without bound (Proposition~\ref{prop:influence}); in simulation, grouping by fitted risk, not grouping as such, protects.
\item \emph{Power.} Among the protected tests, $\EF$ was the most powerful in all 34 simulated settings with $A<0$ and a visible difference, with gains of up to 0.20; the sign of $A$ named the more powerful test in 31 of 33 settings, and Fig.~\ref{fig:map} gives it for asymmetric links in advance.
\end{enumerate}

The practical message is short. When gross recording errors cannot be excluded, group by fitted risk; with ten groups, use $\EF$ where Fig.~\ref{fig:map} predicts $A<0$; with large samples, use groups of about 25 records and the normal references of Theorem~\ref{thm:sparse-null}, not $\chi^2_{G-2}$.

Recent work in this journal studies goodness of fit for discrete-response models through the probability generating function \citep{Meintanis2025} and for binary longitudinal data \citep{Xu2026}; the collection is introduced by \citet{MeintanisJG2025}. Wider comparisons of tests for logistic regression are in \citet{Hosmer1997}, \citet{Kuss2002}, \citet{Liu2024} and \citet{Ebrahim2026bench}, and a related directed test projects the grouped residuals onto a basis of calibration shapes \citep{EbrahimEDGE}.

\begin{figure}[t]
\centering
\includegraphics[width=\textwidth]{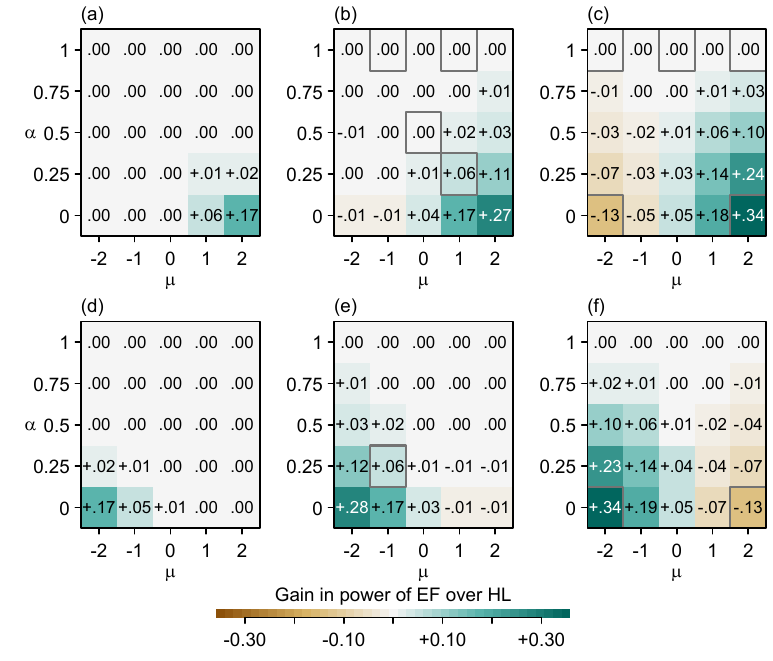}
\caption{Predicted second-order gain in power of $\EF$ over $\HL$ at $n=1000$ (Theorem~\ref{thm:second}) for true links in the Aranda-Ordaz family with parameter $\alpha$, where $\alpha\to0$ is the complementary log-log (panels a--c, spread $s=0.5$, 1, 2 of the linear predictor), and in its mirror image, where $\alpha\to0$ is the log-log (panels d--f, same spreads), by the location $\mu$ of the linear predictor. The fitted model is the linear logistic model; the row $\alpha=1$ is the logit itself. Outlined: the 12 cells simulated to check the map, chosen before the run: five logit controls, one near-zero cell, two moderate cells and the four extreme corners (simulated gains in Online Resource~1, Table~S12). Leading zeros are omitted inside the cells}\label{fig:map}
\end{figure}

Sections~\ref{sec:test}--\ref{sec:sparse} give the statistic and its theory, Section~\ref{sec:sim} the simulations, Section~\ref{sec:data} two applications and Section~\ref{sec:discussion} practical guidance; proofs are in the Appendix and in Online Resource~1.

\section{The grouped Farrington statistic}\label{sec:test}

\subsection{Model, groups and the Hosmer--Lemeshow statistic}\label{sec:model}

Let $y_1,\dots,y_n$ be independent binary responses with covariate vectors $\bm x_i\in\mathbb{R}^p$ that include an intercept. The null hypothesis is the logistic model $\pi_i=\Pr(y_i=1\mid\bm x_i)=\operatorname{expit}(\bm x_i^\top\bm\beta)$. Let $\hat{\bm\beta}$ be the maximum likelihood estimate, $\hat\pi_i$ the fitted risks and $\hat v_i=\hat\pi_i(1-\hat\pi_i)$. Sort the records by $\hat\pi_i$ and cut them into $G$ groups of sizes that differ by at most one. For group $g$ let $n_g$ be its size, $o_g=\sum_{i\in g}y_i$ and $e_g=\sum_{i\in g}\hat\pi_i=n_g\pib_g$ the observed and expected events, and $V_g=n_g\pib_g(1-\pib_g)$. With the standardized group residuals
\begin{equation}\label{eq:r}
r_g=\frac{o_g-e_g}{\sqrt{V_g}},\qquad g=1,\dots,G,
\end{equation}
the Hosmer--Lemeshow statistic is $\HL=\sum_g r_g^2$, referred to the chi-squared distribution on $G-2$ degrees of freedom \citep{Hosmer1980}.

Records with tied fitted risks must not be ordered by the response. With binomial data expanded events first, a group boundary inside a covariate pattern manufactures misfit: in the beetle data of Section~\ref{sec:beetle} this gave $\HL=48.6$ ($p<10^{-7}$), against $p$-values from 0.006 to 0.78 over random orders of the tied records. We break ties at random, or keep each pattern in one group.

\subsection{Farrington's statistic for binary data}\label{sec:farrington}

For binomial counts $y_j$ out of $m_j$ trials at covariate pattern $j$, \citet{Farrington1996} proposed
\[
X^2_\star=\sum_j\frac{(y_j-m_j\hat\pi_j)^2}{m_j\hat\pi_j(1-\hat\pi_j)}-\sum_j\frac{(1-2\hat\pi_j)(y_j-m_j\hat\pi_j)}{m_j\hat\pi_j(1-\hat\pi_j)} .
\]
The weights $a_j=-V'(\mu_j)/V(\mu_j)$, with $V$ the binomial variance function and $\mu_j=m_j\pi_j$, make the statistic locally orthogonal to $\hat{\bm\beta}$ and of minimum variance among first-order corrections; Farrington refers its standardized version to the normal distribution.

\begin{lemma}\label{lem:degenerate}
If $m_j=1$ for every $j$, then $X^2_\star=n$ for every data set and every fitted model.
\end{lemma}

Each record contributes exactly one (Appendix), which is the deficiency that \citet{Kuss2002} found by simulation.

\subsection{The grouped correction and its decomposition}\label{sec:grouping}

Treat the risk groups as covariate patterns with $m_g=n_g$ and fitted risk $\pib_g$. Farrington's correction becomes $C=\sum_g b_g r_g$ with
\begin{equation}\label{eq:b}
b_g=\frac{1-2\pib_g}{\sqrt{V_g}},
\end{equation}
and the grouped Farrington statistic is
\begin{equation}\label{eq:EF}
\EF=\HL-C=\sum_{g=1}^G r_g^2-\sum_{g=1}^G b_g r_g .
\end{equation}
Like $\HL$, it is referred to the chi-squared distribution on $G-2$ degrees of freedom; Theorem~\ref{thm:null} justifies this reference. The following identity explains what the correction does.

\begin{lemma}\label{lem:decomp}
For every partition, every set of fitted risks and every $\bm y\in\{0,1\}^n$, with $\tilde\varepsilon_i=y_i-\pib_g$ for $i\in g$,
\[
\HL=G+C+Q,\qquad \EF=G+Q,\qquad Q=\sum_g\frac{1}{V_g}\sum_{i\neq j\in g}\tilde\varepsilon_i\tilde\varepsilon_j .
\]
\end{lemma}

Each squared group residual splits into single-record and pair terms, and for binary responses the single-record part equals one plus the group's share of $C$ (Appendix~\ref{app:proofs}). $\EF$ thus keeps only how the residuals of distinct records in a group move together, and this pair statistic has an exact optimality.

\begin{proposition}\label{prop:score}
Suppose the records of risk group $g$ are exchangeable binary variables with mean $\pib_g$ and common pairwise correlation $\rho$, the correlated-binomial model of \citet{KupperHaseman1978} applied to the Hosmer--Lemeshow table. The score statistic for $\rho$ at $\rho=0$ is $S=\tfrac12\sum_gS_g$, $S_g=\sum_{i\neq j\in g}z_iz_j$, with $z_i=(y_i-\pib_g)/\{\pib_g(1-\pib_g)\}^{1/2}$, and
\[
\EF-G=\sum_g\frac{S_g}{n_g}\qquad\text{for every data set};
\]
with groups of equal size $m$, $\EF-G=2S/m$. Under $\rho=0$, $S$ is uncorrelated with the score for $\bm\beta$, so it is also the $C(\alpha)$ statistic with $\bm\beta$ estimated, and the test that rejects for large $\EF$ is asymptotically locally most powerful against $\rho>0$.
\end{proposition}

Within narrow risk groups, a smooth miscalibration with bias $h$ makes the residuals of records with similar risk move together, $\mathrm{E}(\tilde\varepsilon_i\tilde\varepsilon_j)\approx\bar h_g^2\ge0$: a positive within-group correlation of the Kupper--Haseman form, whatever the sign of $h$. $\EF$ is the optimal detector of that signature; it has the form of the extra-binomial score tests of \citet{Tarone1979} and \citet{Dean1992}, with clusters formed by fitted risk. $\HL$ contains the same pair statistic and is nearly as powerful against such a departure; the proposition explains why risk-grouped statistics outperform directed tests there, and $\EF$ is the member with the smallest variance (Theorem~\ref{thm:sparse-null}). The removed part $C$ carries the sign of the misfit instead; with one record per group it is $\sum_i\zeta(\hat\pi_i)(y_i-\hat\pi_i)$, $\zeta(\pi)=(1-2\pi)/\{\pi(1-\pi)\}$, the statistic behind the normal tests of \citet{Osius1992} and \citet{Copas1989}, whose weights grow without bound at extreme risks. The test is implemented in the R package \pkg{ebrahim.gof} \citep{EbrahimGof} as \texttt{ef.gof()}.

\subsection{The Pigeon--Heyse statistic}\label{sec:ph}

\citet{PigeonHeyse1999} replaced $V_g$ in $\HL$ by $\Lambda_g=\sum_{i\in g}\hat v_i$, which accounts for the spread of the fitted risks within a group, and referred the result to $\chi^2_{G-1}$. The next proposition places it beside $\HL$.

\begin{proposition}\label{prop:ph}
Let $J^2=\sum_g(o_g-e_g)^2/\Lambda_g$. For every data set, $V_g-\Lambda_g=\sum_{i\in g}(\hat\pi_i-\pib_g)^2$ and
\[
\HL\le J^2\le\HL\cdot\max_g(V_g/\Lambda_g).
\]
Under the conditions of Theorem~\ref{thm:null}, $J^2\to\sum_j\mu_j\chi^2_{1,j}$ in distribution, with $\mu_j$ the eigenvalues of the limit of $\bm I-\bm P_{\!\Lambda}$, $\bm P_{\!\Lambda}=\bm\Lambda^{-1/2}\bm U\bm M^{-1}\bm U^\top\bm\Lambda^{-1/2}$ (notation of Section~\ref{sec:setting}). The matrix $\bm P_{\!\Lambda}$ has eigenvalue one with eigenvector $\bm\Lambda^{1/2}\bm 1_G$, so one $\mu_j$ is exactly zero and $\sum_j\mu_j=G-\operatorname{tr}\bm P_{\!\Lambda}\le G-1$.
\end{proposition}

The Pigeon--Heyse statistic is thus $\HL$ enlarged by the within-group spread of the fitted risks and referred to a distribution with one degree of freedom too many, because $\bm P_{\!\Lambda}$ has a second eigenvalue near one, the slope direction. With ten groups the eigenvalue sum was 7.86 to 8.03 (Online Resource~1, Table~S7), so the $\chi^2_{G-1}$ reference gives an asymptotic size of 0.028 to 0.032; after size adjustment the two tests differed in power by 0.001 on average.

\subsection{A directional read-out}\label{sec:readout}

The terms $c_g=b_gr_g$ of $C$ describe the direction of the misfit. Since $b_g$ changes sign at $\pib_g=1/2$, $C<0$ says that, on balance, the model underestimates how fast the risk approaches one in the high-risk groups or approaches zero in the low-risk groups. The read-out is descriptive (Section~\ref{sec:beetle}).

\section{Theory with a fixed number of groups}\label{sec:theory}

\subsection{Setting}\label{sec:setting}

Write $\bm r=(r_1,\dots,r_G)^\top$, $\bm b=(b_1,\dots,b_G)^\top$, $\bm\Lambda=\operatorname{diag}(\Lambda_g)$, $\bm V=\operatorname{diag}(V_g)$, $\bm U$ the $G\times p$ matrix with rows $\sum_{i\in g}\hat v_i\bm x_i^\top$ and $\bm M=\sum_i\hat v_i\bm x_i\bm x_i^\top$; matrices are set in bold. We assume: (A1) the $\bm x_i$ are independent and identically distributed with $\mathrm{E}\|\bm x\|^4<\infty$ and $\mathrm{E}(\bm x\bm x^\top)$ positive definite, and $\bm\beta$ is fixed; (A2) the distribution of $\bm x^\top\bm\beta$ has a continuous density, so that the group boundaries converge to quantiles of the fitted risk; (A3) $G\ge 3$ is fixed. Under (A1)--(A3) the maximum likelihood estimate admits the expansion $\hat{\bm\beta}-\bm\beta=\bm M^{-1}\sum_i\bm x_i(y_i-\pi_i)+o_p(n^{-1/2})$, and the groups formed at estimated quantiles of the fitted risk behave to first order as groups formed at the limiting quantiles. This follows from an empirical-process argument for data-dependent cells \citep{Pollard1979}, given in the proof of Theorem~\ref{thm:null}, and is the step that \citet{Hosmer1980} took from \citet{MooreSpruill1975}; \citet{SurjanovicLoughin2023} show how its size falls as the model grows. Let
\begin{equation}\label{eq:Sigma}
\bm\Sigma=\bm V^{-1/2}\bigl(\bm\Lambda-\bm U\bm M^{-1}\bm U^\top\bigr)\bm V^{-1/2}
\end{equation}
with eigenvalues $\lambda_1\ge\dots\ge\lambda_G\ge0$ and limit $\bm\Sigma_\infty$. Because $\Lambda_g\le V_g$ by Jensen's inequality and $\bm U\bm M^{-1}\bm U^\top$ is positive semidefinite, every $\lambda_j\le1$. Since $\bm x_i$ contains the intercept, $\bm U^\top\bm 1_G=\bm M\bm e_1$, so $(\bm\Lambda-\bm U\bm M^{-1}\bm U^\top)\bm 1_G=\bm 0$ and $\bm\Sigma\bm V^{1/2}\bm 1_G=\bm 0$: $\lambda_G=0$ exactly, which reflects the identity $\sum_g(o_g-e_g)=0$ of the intercept's score equation.

\subsection{The null limit}\label{sec:null}

\begin{theorem}\label{thm:null}
Under the null hypothesis and (A1)--(A3), $\bm r\to N_G(\bm 0,\bm\Sigma_\infty)$ in distribution and $C=O_p(n^{-1/2})$. Hence $\HL$ and $\EF$ have the same limiting distribution, that of $\sum_j\lambda_j\chi^2_{1,j}$ with independent $\chi^2_{1,j}$ and $\lambda_j$ the eigenvalues of $\bm\Sigma_\infty$, and $\EF-\HL\to0$ in probability. The limit is close to $\chi^2_{G-2}$ when $G-2$ eigenvalues are close to one and two are close to zero; one of them, $\lambda_G$, is exactly zero.
\end{theorem}

For $\HL$, Theorem~\ref{thm:null} restates the limit of \citet{Hosmer1980} in the rigorous form of \citet{Surjanovic2023}; what it adds is that $C=O_p(n^{-1/2})$, so that $\EF$ shares that limit and the same first-order local power. What separates the two tests is the term of order $n^{-1/2}$ of Theorem~\ref{thm:second}. In the five designs of Section~\ref{sec:sim} the mean eigenvalue sum was 7.57 to 8.00 for $G=10$, close to $G-2=8$.

\subsection{A Gaussian representation}\label{sec:gauss}

\begin{proposition}\label{thm:gauss}
Suppose $\bm r\sim N_G(\bm m,\bm\Sigma)$ with $\bm\Sigma=\bm\Gamma\operatorname{diag}(\lambda_j)\bm\Gamma^\top$, $\bm\Gamma$ a $G\times G$ orthogonal matrix, and let $\bm c=\bm\Gamma^\top\bm b$ and $\bm\nu=\bm\Gamma^\top\bm m$. Then
\begin{equation}\label{eq:law}
\EF\overset{d}{=}\sum_{j:\lambda_j>0}\lambda_j\,\chi^2_1\!\left(\frac{(2\nu_j-c_j)^2}{4\lambda_j}\right)-\frac14\sum_{j:\lambda_j>0}c_j^2+\sum_{j:\lambda_j=0}\bigl(\nu_j^2-c_j\nu_j\bigr),
\end{equation}
with independent noncentral chi-squared variables, and
\begin{equation}\label{eq:alignment}
\mathrm{E}(\EF)-\mathrm{E}(\HL)=-A,\qquad A=\bm b^\top\bm m=\sum_{g=1}^G\frac{(1-2\pib_g)\,\delta_g}{V_g},\quad \delta_g=\mathrm{E}(o_g-e_g).
\end{equation}
\end{proposition}

Under the null, \eqref{eq:law} is a weighted noncentral chi-squared variable minus a constant, computable by Imhof's formula \citep{Imhof1961,Duchesne2010}; its terms are of order $n^{-1}$, and it is not a better reference than $\chi^2_{G-2}$ (Section~\ref{sec:size}). The alignment $A$ is the inner product of the bias of the group residuals with the correction weights: with $d_g$ the mean bias of the fitted risks in group $g$, $A=\sum_g(1-2\pib_g)d_g/\{\pib_g(1-\pib_g)\}$. Under local departures $d_g=n^{-1/2}h_g$, $A=O(n^{-1/2})$; in Section~\ref{sec:power} the theory is evaluated at the finite-sample $\bm m$ and $A$ of each departure. A symmetric link misfitted by the logit also gives $A\neq0$: for a probit truth with $\eta\sim N(0,s^2)$, $A=-0.11$, $-1.55$ and $-3.30$ at $s=1$, 2 and 3, because the logistic fit is less extreme than the truth at both ends.

\subsection{Second-order power}\label{sec:second}

Let $q$ be a common critical value of the two tests at level $\alpha$.

\begin{theorem}\label{thm:second}
Let $\bm b=n^{-1/2}\tilde{\bm b}$ with $\tilde{\bm b}$ bounded, and let $\HL=\|\bm r\|^2$ and $C=\bm b^\top\bm r$ satisfy: (i) $\sup_n\mathrm{E}\|\bm r\|^2<\infty$ (local alternatives, $\bm m=O(1)$); (ii) $(\HL,C)$ has a joint density, $f(q)>0$ for the density $f$ of $\HL$, and the conditional density $f_{\HL\mid C}(x\mid l)$ is differentiable in $x$ with $|\partial_xf_{\HL\mid C}(x\mid l)|\le K$ for $|x-q|\le\varepsilon$ and all $l$. Then
\begin{equation}\label{eq:second}
\Pr(\EF>q)-\Pr(\HL>q)=-f(q)\,\mathrm{E}(C\mid\HL=q)+O(n^{-1}) .
\end{equation}
(a) If $\bm r$ is symmetric about zero, $\mathrm{E}(C\mid\HL=q)=0$, so the two tests have the same size to order $n^{-1}$.
(b) If $\bm r\sim N_G(\bm m,\bm\Sigma)$ and $\bm\Sigma$ is a projection of rank $k$ with $\bm m$ in its range, then $\mathrm{E}(C\mid\HL=q)=\rho\,A$ with
\begin{equation}\label{eq:rho}
\rho=\frac{\sqrt q}{\|\bm m\|}\,\frac{I_{k/2}(\|\bm m\|\sqrt q)}{I_{k/2-1}(\|\bm m\|\sqrt q)}>0 ,
\end{equation}
where $I_\nu$ is the modified Bessel function of the first kind, and $f$ is the noncentral $\chi^2_k(\|\bm m\|^2)$ density.
\end{theorem}

\begin{corollary}\label{cor:sign}
Under the conditions of Theorem~\ref{thm:second}(b), if $\sqrt n\,A\to a\neq0$, then for all sufficiently large $n$ the grouped Farrington test is more powerful than the Hosmer--Lemeshow test if and only if $a<0$, and the gain is $-\rho f(q)A+O(n^{-1})$.
\end{corollary}

The conditions hold for the Gaussian model of Proposition~\ref{thm:gauss}; for the discrete binary statistics the theorem is an approximation through that model, which Section~\ref{sec:second-check} checks. Part (a) applies to the extent that the group residuals are nearly symmetric, which fails in nearly empty groups (Section~\ref{sec:size}). In the simulations each test is compared with its own null 95\% point; this size adjustment removes the null-size difference of order $n^{-1/2}$, and raw rates are reported beside it.

\subsection{One corrupted record}\label{sec:influence}

\begin{proposition}\label{prop:influence}
Fix the coefficients at $\hat{\bm\beta}$ and replace one record $(\bm x_i,y_i)$ by an arbitrary $(\tilde{\bm x},\tilde y)$, with the groups re-formed. If $V_{\min}=\min_g V_g\ge2$, then
\[
|\Delta\HL|\le 4\sqrt2\,G\,(1+\|\bm r\|_\infty)^2\,V_{\min}^{-1/2},\qquad |\Delta\EF|\le 11\,G\,(1+\|\bm r\|_\infty)^2\,V_{\min}^{-1/2},
\]
with $\bm r$ computed on the original data. Stukel's joint score statistic has no such bound: its change tends to infinity as $|\tilde{\bm x}^\top\hat{\bm\beta}|\to\infty$ for a record whose response is the less likely one ($\tilde y=0$ as $\tilde{\bm x}^\top\hat{\bm\beta}\to+\infty$, or $\tilde y=1$ as $\tilde{\bm x}^\top\hat{\bm\beta}\to-\infty$).
\end{proposition}

The bound does not depend on $\tilde{\bm x}$: a record enters the grouped statistics only through $\tilde y-\tilde\pi\in(-1,1)$, through $\tilde\pi$ and through its rank, so its effect is $O_p(n^{-1/2})$. The bound is conservative (the largest change over 3000 adversarial data sets was 0.62 of it). It holds the coefficients fixed; a record far outside the covariate range moves the fit itself, and then every test that uses the fit moves.

\section{Theory with many small groups}\label{sec:sparse}

With ten groups the correction is a second-order effect. It becomes first order when the number of groups grows with $n$ and each group holds a bounded number $m$ of records: the sparse regime for which Farrington designed his correction, and the regime of guidelines that let $G$ grow with $n$ while keeping the $\chi^2_{G-2}$ reference \citep{Paul2013,Nattino2020}.

Moment corrections for Pearson-type statistics with sparse binomial data go back to \citet{McCullagh1985}, \citet{Osius1992}, \citet{Farrington1996} and \citet{PaulDeng2000}; they take the covariate patterns as given. We assume (S1) non-random covariates with $\max_i\|\bm x_i\|\le K$, the intercept as first component, and $n^{-1}\bm M$ with eigenvalues in a fixed compact subset of $(0,\infty)$ near $\bm\beta$; (S2) a partition that does not depend on $\bm y$, with group sizes $m\le n_g<2m$, $m\ge2$ and $m/n\to0$; (S3) groups narrow in risk, $\varrho_n=\max_g\max_{i,j\in g}|\pi_i-\pi_j|\to0$. Here $\pi_i$, $v_i=\pi_i(1-\pi_i)$, $\pib_g$ and $V_g$ are taken at the true parameter. Let $\tau_g=\sum_{i\in g}v_i$, $\zeta(\pi)=(1-2\pi)/\{\pi(1-\pi)\}$, and
\[
\mu_n=\sum_g\frac{\tau_g}{V_g},\qquad \sigma_Q^2=2\sum_g\frac{\tau_g^2-\sum_{i\in g}v_i^2}{V_g^2},\qquad \sigma_L^2=\sum_g\frac{(1-2\pib_g)^2\tau_g}{V_g^2}-\bm t^\top\bm M^{-1}\bm t,
\]
with $\bm t=\sum_g(1-2\pib_g)\bm u_g/V_g$, $\bm u_g=\sum_{i\in g}v_i\bm x_i$, and $\sigma_{\HL}^2=\sigma_Q^2+\sigma_L^2$. Let $R=R_n$ be the weighted residual mean square of $\zeta(\pi_i)$ regressed on $\bm x_i$ with weights $v_i$, a property of the covariates at hand.

\begin{theorem}\label{thm:sparse-null}
Under (S1)--(S3) and the null hypothesis, as $n\to\infty$:
(i) $(\EF-\mu_n)/\sigma_Q\to N(0,1)$ and $(\HL-\mu_n)/\sigma_{\HL}\to N(0,1)$;
(ii) $\sigma_Q^2=2\sum_g(1-1/n_g)\{1+o(1)\}$, which does not depend on the model, and with equal group sizes $\sigma_{\HL}^2/\sigma_Q^2\to1+R/\{2(m-1)\}$;
(iii) if $\liminf R>0$, $(\EF-\mu_n)/\sigma_Q$ and $C/\sigma_L$ are asymptotically independent standard normal variables;
(iv) the conclusions hold with the moments replaced by their values at $\hat{\bm\beta}$.
Moreover, among all statistics $\HL-\ell$ with $\ell=\sum_i\alpha_i(y_i-\hat\pi_i)$ and $\max_i|\alpha_i|=O(1/m)$, $\EF$ has the smallest asymptotic variance, and its first-order expansion does not involve $\hat{\bm\beta}-\bm\beta$. If, in addition, the group sizes are equal and $\sqrt{G}\,\varrho_n^2\to0$, which holds for groups formed by risk when the linear predictor has a bounded density, then $(\EF-G)/\sqrt{2G}\to N(0,1-1/m)$ and $(\HL-G)/\sqrt{2G}\to N(0,1-1/m+R/(2m))$ for $R_n\to R$, so the $\chi^2_{G-2}$ reference is asymptotically conservative for $\EF$ for every fixed $m$; for $\HL$ it is conservative if $R<2$ and liberal if $R>2$.
\end{theorem}

The variance reduction is Farrington's minimum-variance property carried over to risk groups, and it settles why the correction enters with weight one: among $\HL-\kappa C$ the variance is smallest at $\kappa=1$. The reduction is small for a narrow risk range ($R_n=0.05$ in design D1) and large for a wide one (a variance ratio of 1.9 at $m=2$ with $R_n=1.79$; Online Resource~1, Table~S3). By part (iii), $C/\sigma_L$ is to first order the Osius--Rojek statistic \citep{Osius1992}.

\begin{theorem}\label{thm:sparse-power}
Let $y_i$ be independent Bernoulli$(\pi_i^\dagger)$ with $\pi_i^\dagger=\pi_i(\bm\beta_n^\ast)+h_{ni}$, where $\bm\beta_n^\ast$ solves $\sum_i\bm x_i\{\pi_i^\dagger-\pi_i(\bm\beta)\}=\bm 0$; assume (S1)--(S3) at $\bm\beta_n^\ast$, $\max_i|h_{ni}|\to0$, $m\max_ih_{ni}^2\to0$ and $\sqrt{G}\,\varrho_n\max_i|h_{ni}|\to0$. Let $\Delta_n=\sum_gV_g^{-1}\sum_{i\neq j\in g}h_{ni}h_{nj}$ and $A_n=\sum_g(1-2\pib_g)\delta_g/V_g$ with $\delta_g=\sum_{i\in g}h_{ni}$, and suppose $\Delta_n/\sigma_Q\to\theta_Q$, $A_n/\sigma_{\HL}\to\vartheta$ and $\sigma_Q/\sigma_{\HL}\to\omega$. Then $Z_{\EF}=(\EF-\hat\mu_n)/\hat\sigma_Q\to N(\theta_Q,1)$ and $Z_{\HL}=(\HL-\hat\mu_n)/\hat\sigma_{\HL}\to N(\omega\theta_Q+\vartheta,1)$. Consequently the one-sided $\EF$ test is asymptotically more powerful than the one-sided $\HL$ test if and only if
\[
A_n<\Delta_n\bigl(\sigma_{\HL}/\sigma_Q-1\bigr)\{1+o(1)\},
\]
which holds whenever $A_n<0<\Delta_n$, and also when $A_n=0<\Delta_n$ if $\liminf R_n>0$; the $\HL$ test is asymptotically biased, with power below its size, whenever $A_n<-\Delta_n$.
\end{theorem}

$A_n$ is the alignment of Proposition~\ref{thm:gauss}, now acting at first order, and for smooth departures the efficacy of $\EF$ grows like $\sqrt{m-1}$. The proofs, in Online Resource~1, extend the moment theory of \citet{McCullagh1985}, \citet{Osius1992} and \citet{Farrington1996} from given binomial patterns to constructed risk groups through a Lyapunov central limit theorem for within-group pair sums \citep{Morris1975,FahrmeirKaufmann1985}. Condition (S2) excludes grouping on $\hat\pi$. Online Resource~1 (Proposition~S1) proves that the theorems also hold with groups formed on the fitted risk when there is one continuous covariate, for any group size; when the groups are large ($m\gg n^{1/2}\log n$); and for any group size when fitted-risk bins are split at random. In the remaining case, blocks of a few records ranked on several covariates, grouping on the fitted and on the true risk gave the same standard deviations and sizes in simulation (Tables~S1 and S2).

The regime matters in practice. For $1000<n\le25000$, \citet{Paul2013} recommend $G=\max\{10,\min(n_1/2,\,(n-n_1)/2,\,2+8(n/1000)^2)\}$ groups, $n_1$ the number of events, with the $\chi^2_{G-2}$ reference \citep{Hosmer2013}: about 25 records per group at $n=5000$, 12 at $n=10000$ and 6 at $n=20000$, exactly the setting of Theorem~\ref{thm:sparse-null} (Section~\ref{sec:large-sim}). The rule is designed to hold the power near its value at $n=1000$; Theorem~\ref{thm:sparse-power} shows the price, since for a smooth misfit the pair drift $\Delta_n$ does not grow with $G$ while the null standard deviation grows like $(2G)^{1/2}$.

\begin{remark}\label{rem:or}
The removed part $C/\sigma_L$ is a powerful test against a wrong link on clean data, but its unbounded weights make it as fragile as a record-level score: with ten reversed covariates in 4000 records and groups of $m=25$, its false-alarm rate was 0.41, against 0.10 for $\EF$ (Online Resource~1, Table~S5). Protection also needs groups that carry several expected events, as Proposition~\ref{prop:influence} shows, and with groups of a few records formed on the fitted risk a single response re-forms the groups and moves $\EF$ by about $m^{-1/2}$ standard deviations (Online Resource~1, Remark~S1); we use groups of about 25 records.
\end{remark}

\section{Simulation study}\label{sec:sim}

\subsection{Design}\label{sec:design}

Five designs fix the covariates and the linear predictor $\eta$; the working model is always the linear logistic model in the covariates, fitted by maximum likelihood.
\begin{itemize}
\item D1 (balanced): $x_1\sim U(-3,3)$, $x_2\sim\mathrm{Bernoulli}(0.5)$, $\eta=0.6x_1+0.5x_2$.
\item D2 (rare events, about 10\%): as D1 with $\eta=-2.8+0.6x_1+0.5x_2$.
\item D3 (skewed covariate): $x_1\sim\mathrm{Exp}(1)-1$, $\eta=-0.3+0.9x_1+0.5x_2$.
\item D4 (strong discrimination, area under the ROC curve 0.93): as D1 with $\eta=1.6x_1+0.8x_2$.
\item D5 (six covariates): four standard normal and two Bernoulli(0.3) covariates, $\eta=0.4\sum_kx_k$.
\end{itemize}
Under the null the response is logistic in $\eta$. The departures change the link or the predictor: complementary log-log, $\pi=1-\exp(-e^{\eta})$; log-log, $\pi=\exp(-e^{-\eta})$; probit, $\pi=\Phi(\eta)$; an omitted square, $\operatorname{expit}(\eta+0.15x_1^2)$; and an omitted interaction, $\operatorname{expit}(\eta+0.5x_1x_2)$. A change of link also changes the event rate (in D2 the log-log link lowers it from 10\% to 1.4\%), so in the main power study the intercept of each departure is shifted to keep the event rate of the null (Section~\ref{sec:power}).

The study is a funnel. Stage~1 screens nine tests for their level under a correct model with corrupted records: three that group by fitted risk ($\EF$, $\HL$ and Pigeon--Heyse), two that group by covariate clusters \citep{Tsiatis1980,Xie2008}, and four record-level directed tests: Stukel's joint score test \citep{Stukel1988}, the cubic calibration likelihood-ratio test, the GiViTI test \citep{Nattino2014} and Stukel-W, a bounded-leverage version of Stukel's test that we add as the natural robust rival, whose added directions use the fitted linear predictor winsorized at its 2.5\% and 97.5\% quantiles. Stage~2 compares the power of the tests that pass. $\EF$ and $\HL$ use $G=10$ and $\chi^2_8$ unless stated, with ties broken at random; the other tests use \pkg{ebrahim.gof} \citep{EbrahimGof} with its default references. Settings have 2000 replications (Monte Carlo standard error 0.005 at 0.05) unless stated. Size-adjusted power uses each statistic's 95\% point in the matched null setting, and raw rates are reported beside it.

Every study was declared, with its claims, before it was run and fixed by its SHA-256 hash. The declarations are deposited with the results \citep{EbrahimDeposit}, the last three publicly before their runs \citep{EbrahimPrereg,EbrahimPrereg4,EbrahimPrereg5}, and Online Resource~1, Table~S13 lists each claim with its outcome.

\subsection{Stage 1: level under corrupted records}\label{sec:screen}

Under a correct model with $n=1000$, $k$ records had a recording error before the fit: $x_1$ multiplied by 4 or by $-4$, which places most corrupted values outside the covariate range; a sign error $x_1\to-x_1$, which stays inside it; or a misclassified response $y\to1-y$.

\begin{table}[t]
\caption{Stage 1, design D1: false-alarm rate at the 5\% level under a correct model with $k$ corrupted records in $n=1000$ (2000 replications; Monte Carlo standard error about 0.005 at 0.05). Errors: $x_1$ multiplied by 4 or by $-4$; sign error $x_1\to-x_1$; response misclassified. Stukel-W: Stukel's test with its added directions built from the winsorized linear predictor. Bold: above 0.10}\label{tab:screen}
\centering\scriptsize\setlength{\tabcolsep}{2.2pt}
\begin{tabular}{@{}lrrrrrrrrrr@{}}
\toprule
 & & \multicolumn{3}{c}{risk-ordered groups} & \multicolumn{2}{c}{covariate clusters} & \multicolumn{4}{c}{record-level directed} \\
\cmidrule(lr){3-5}\cmidrule(lr){6-7}\cmidrule(l){8-11}
Error & $k$ & EF & HL & P--H & Tsiatis & Xie & Stukel & Stukel-W & Cubic & GiViTI \\
\midrule
none & 0 & 0.052 & 0.052 & 0.028 & 0.048 & 0.026 & 0.055 & 0.054 & 0.053 & 0.054 \\
\addlinespace
$\times4$ & 1 & 0.046 & 0.044 & 0.029 & 0.089 & 0.086 & \textbf{0.134} & 0.056 & \textbf{0.141} & 0.098 \\
 & 2 & 0.059 & 0.059 & 0.039 & \textbf{0.149} & \textbf{0.162} & \textbf{0.227} & 0.066 & \textbf{0.206} & \textbf{0.152} \\
 & 5 & 0.050 & 0.048 & 0.030 & \textbf{0.296} & \textbf{0.335} & \textbf{0.393} & 0.074 & \textbf{0.322} & \textbf{0.254} \\
 & 10 & 0.054 & 0.052 & 0.037 & \textbf{0.554} & \textbf{0.622} & \textbf{0.619} & \textbf{0.126} & \textbf{0.468} & \textbf{0.387} \\
\addlinespace
$\times(-4)$ & 1 & 0.050 & 0.050 & 0.030 & \textbf{0.177} & \textbf{0.230} & \textbf{0.365} & 0.067 & \textbf{0.392} & \textbf{0.249} \\
 & 2 & 0.056 & 0.054 & 0.035 & \textbf{0.342} & \textbf{0.418} & \textbf{0.595} & 0.090 & \textbf{0.604} & \textbf{0.460} \\
 & 5 & 0.055 & 0.052 & 0.030 & \textbf{0.712} & \textbf{0.804} & \textbf{0.902} & \textbf{0.262} & \textbf{0.893} & \textbf{0.808} \\
 & 10 & \textbf{0.112} & \textbf{0.104} & 0.074 & \textbf{0.952} & \textbf{0.979} & \textbf{0.991} & \textbf{0.533} & \textbf{0.984} & \textbf{0.965} \\
\addlinespace
sign & 1 & 0.045 & 0.046 & 0.027 & 0.049 & 0.028 & 0.046 & 0.046 & 0.046 & 0.052 \\
 & 2 & 0.047 & 0.048 & 0.031 & 0.058 & 0.028 & 0.058 & 0.056 & 0.056 & 0.048 \\
 & 5 & 0.043 & 0.044 & 0.025 & 0.050 & 0.032 & 0.056 & 0.055 & 0.058 & 0.050 \\
 & 10 & 0.037 & 0.036 & 0.024 & 0.053 & 0.030 & 0.043 & 0.042 & 0.042 & 0.044 \\
 & 20 & 0.037 & 0.037 & 0.022 & 0.050 & 0.036 & 0.054 & 0.050 & 0.046 & 0.056 \\
\addlinespace
$y$ flipped & 5 & 0.050 & 0.050 & 0.036 & 0.049 & 0.030 & 0.048 & 0.053 & 0.046 & 0.050 \\
 & 10 & 0.048 & 0.048 & 0.030 & 0.050 & 0.032 & 0.051 & 0.048 & 0.044 & 0.054 \\
 & 20 & 0.056 & 0.056 & 0.036 & 0.050 & 0.032 & 0.058 & 0.056 & 0.056 & 0.053 \\
 & 50 & 0.049 & 0.050 & 0.032 & 0.052 & 0.030 & 0.053 & 0.052 & 0.048 & 0.052 \\
\bottomrule
\end{tabular}

\end{table}

\begin{figure}[t]
\centering
\includegraphics[width=\textwidth]{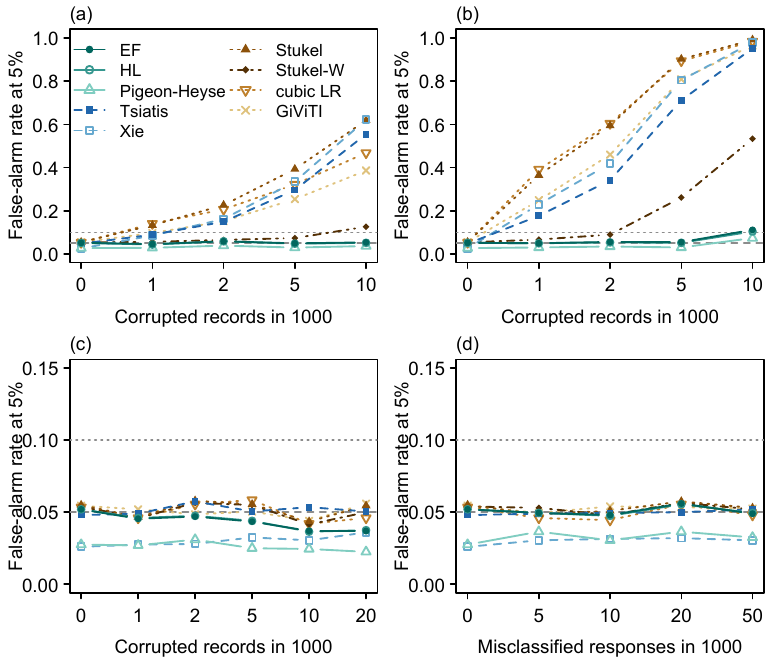}
\caption{Stage 1 in design D1: false-alarm rate at the 5\% level under a correct model with $k$ corrupted records in 1000, for (a) a covariate multiplied by 4, (b) multiplied by $-4$, (c) a sign error inside the covariate range and (d) misclassified responses; panels (c) and (d) have a vertical scale to 0.15. Solid lines: tests that group by fitted risk; dashed: covariate-cluster tests; dotted and dash-dotted: record-level directed tests. Dashed horizontal line: 0.05; dotted: 0.10}\label{fig:screen}
\end{figure}

Gross errors break the directed and the covariate-cluster tests (Table~\ref{tab:screen}, Fig.~\ref{fig:screen}). One record multiplied by $-4$ raised the false-alarm rates of Stukel's test, the cubic calibration test and the GiViTI test to 0.37, 0.39 and 0.25, and those of the Tsiatis and Xie tests to 0.18 and 0.23; ten records multiplied by 4 raised the latter to 0.55 and 0.62. With up to five gross errors of either kind, $\EF$ and $\HL$ stayed between 0.044 and 0.059 and the Pigeon--Heyse test below 0.04.

Bounding the leverage helps but does not close the gap: Stukel-W held its level with up to five enlarged records but reached 0.26 with five reversed ones. In the burn-injury covariates of Section~\ref{sec:burn}, with the burned area of $k$ patients multiplied by four and capped at the possible 100\%, Stukel's test reached 0.15 with ten patients, while $\EF$ stayed at or below 0.035 (Online Resource~1, Table~S9, which also gives design D5).

Errors that create no leverage, in-range sign errors and misclassified responses, left all nine tests at their level (0.022 to 0.069). What separates the tests is the gross error, the high-leverage record that the robustness literature treats as the main threat to the level of a test \citep{HeritierRonchetti1994,CantoniRonchetti2001}.

Proposition~\ref{prop:influence} explains the difference. To illustrate it, one non-event at $x_1\approx1.5$ had its covariate multiplied by up to 1000: at the coefficients of the clean fit, $\EF$ changed by at most 0.63 and Stukel's statistic by up to $1.5\times10^{10}$. After refitting, $\EF$ changed by at most 0.43 up to eight times the covariate range; beyond that the record moves the fit itself, as it would for any test. Stage~2 continues with the three risk-grouped tests, with Stukel-W and Tsiatis's test as references.

\subsection{Size of the risk-grouped tests}\label{sec:size}

Online Resource~1, Table~S17, gives the size of $\EF$ and $\HL$ in the five designs. $\HL$ is within three Monte Carlo standard errors of 0.05 in all 25 settings and $\EF$ in 22. $\EF$ is conservative (0.028 to 0.044) only where the smallest expected number of events or non-events in a group, $e_{\min}=\min_g\min(e_g,n_g-e_g)$, has median below 0.5, that is with rare events or strong discrimination at $n\le500$. There a single event in a nearly empty group gives a large residual that $C$ partly absorbs; the effect comes from skewness, so neither recentring $C$ by its second-order mean \citep{CoxSnell1968,Firth1993} nor the finite-sample law \eqref{eq:law} removes it.

\smallskip\noindent\textbf{Rule.} The grouped Farrington test needs about one expected event and one expected non-event in each risk group, much less than the five per cell asked of Pearson-type statistics. When some group expects fewer than about one half, use fewer groups; the Hosmer--Lemeshow reference is then also unreliable, as the plasmode of Section~\ref{sec:burn} shows (size 0.08 to 0.10).

\subsection{Stage 2: power of the risk-grouped tests}\label{sec:power}

In the main power study each departure keeps the event rate of its null, its intercept being shifted accordingly; otherwise the D2 departures lower the event rate to 1.4\% and 3.9\%, and the size adjustment compares data sets of different sparsity. Table~\ref{tab:power} gives the power at $n=1000$ with $A$ and the paired standard error of the difference between $\EF$ and $\HL$; $n=500$, 2000 and the event-rate-free study are in Tables~S10 and S11.

\begin{table}[t]
\caption{Stage 2: power at the 5\% level at $n=1000$ with ten groups, each departure with the event rate of its null (2000 replications). $A$: alignment of the departure (Proposition~\ref{thm:gauss}). Size-adjusted and raw rejection rates of the risk-grouped tests; EF $-$ HL: size-adjusted difference with its paired Monte Carlo standard error; size-adjusted power of Tsiatis's test and Stukel-W for reference. Bold: the most powerful risk-grouped test where EF and HL differ by more than 0.01}\label{tab:power}
\centering\scriptsize\setlength{\tabcolsep}{1.15pt}
\begin{tabular}{@{}llrrrrrrrrr@{}}
\toprule
 & & & \multicolumn{3}{c}{size-adjusted} & & \multicolumn{2}{c}{raw} & \multicolumn{2}{c}{reference (adjusted)} \\
\cmidrule(lr){4-6}\cmidrule(lr){8-9}\cmidrule(l){10-11}
Design & Truth & $A$ & EF & HL & P--H & EF $-$ HL (SE) & EF & HL & Tsiatis & Stukel-W \\
\midrule
D1 & cloglog & $-$0.33 & \textbf{0.424} & 0.400 & 0.403 & 0.024 (0.003) & 0.430 & 0.408 & 0.305 & 0.726 \\
 & log-log & 0.10 & 0.252 & 0.257 & 0.259 & $-$0.006 (0.002) & 0.256 & 0.263 & 0.227 & 0.535 \\
 & probit & $-$0.17 & 0.058 & 0.052 & 0.053 & 0.005 (0.002) & 0.059 & 0.054 & 0.043 & 0.080 \\
 & square & $-$0.12 & 0.888 & 0.885 & 0.885 & 0.003 (0.001) & 0.892 & 0.891 & 0.930 & 0.983 \\
 & interaction & $-$0.01 & 0.077 & 0.076 & 0.076 & 0.002 (0.001) & 0.080 & 0.079 & 0.896 & 0.135 \\
\addlinespace
D2 & cloglog & 0.14 & 0.065 & 0.065 & 0.065 & 0.000 (0.003) & 0.036 & 0.042 & 0.062 & 0.062 \\
 & log-log & $-$4.17 & \textbf{0.207} & 0.042 & 0.043 & 0.165 (0.008) & 0.132 & 0.029 & 0.038 & 0.334 \\
 & probit & $-$2.15 & \textbf{0.084} & 0.026 & 0.026 & 0.058 (0.006) & 0.046 & 0.015 & 0.033 & 0.098 \\
 & square & 2.98 & 0.490 & 0.610 & \textbf{0.611} & $-$0.120 (0.007) & 0.431 & 0.555 & 0.625 & 0.790 \\
 & interaction & 3.14 & 0.230 & \textbf{0.360} & 0.360 & $-$0.130 (0.008) & 0.177 & 0.310 & 0.526 & 0.374 \\
\addlinespace
D3 & cloglog & $-$0.23 & \textbf{0.190} & 0.172 & 0.172 & 0.018 (0.003) & 0.165 & 0.154 & 0.044 & 0.320 \\
 & log-log & 0.08 & 0.130 & 0.130 & 0.134 & $-$0.000 (0.002) & 0.110 & 0.116 & 0.201 & 0.277 \\
 & probit & $-$0.09 & 0.066 & 0.059 & 0.059 & 0.007 (0.002) & 0.053 & 0.052 & 0.025 & 0.048 \\
 & square & $-$0.05 & 0.112 & 0.105 & 0.106 & 0.007 (0.002) & 0.096 & 0.092 & 0.046 & 0.184 \\
 & interaction & $-$0.01 & 0.102 & 0.100 & 0.100 & 0.002 (0.001) & 0.087 & 0.086 & 0.285 & 0.137 \\
\addlinespace
D4 & cloglog & 0.09 & \textbf{0.515} & 0.482 & 0.485 & 0.033 (0.007) & 0.514 & 0.502 & 0.346 & 0.832 \\
 & log-log & 3.11 & 0.503 & \textbf{0.586} & 0.586 & $-$0.083 (0.009) & 0.502 & 0.614 & 0.457 & 0.851 \\
 & probit & $-$3.80 & \textbf{0.037} & 0.006 & 0.007 & 0.031 (0.004) & 0.037 & 0.007 & 0.008 & 0.028 \\
 & square & $-$0.22 & \textbf{0.264} & 0.228 & 0.227 & 0.036 (0.004) & 0.263 & 0.257 & 0.226 & 0.530 \\
 & interaction & 0.18 & 0.045 & 0.047 & 0.046 & $-$0.002 (0.004) & 0.045 & 0.056 & 0.213 & 0.068 \\
\addlinespace
D5 & cloglog & $-$0.16 & \textbf{0.226} & 0.214 & 0.215 & 0.012 (0.002) & 0.202 & 0.193 & 0.070 & 0.433 \\
 & log-log & 0.06 & 0.156 & 0.158 & 0.158 & $-$0.002 (0.002) & 0.140 & 0.144 & 0.053 & 0.311 \\
 & probit & $-$0.06 & 0.059 & 0.057 & 0.057 & 0.002 (0.001) & 0.053 & 0.050 & 0.045 & 0.048 \\
 & square & $-$0.01 & 0.068 & 0.066 & 0.066 & 0.002 (0.001) & 0.058 & 0.058 & 0.087 & 0.075 \\
 & interaction & $-$0.04 & 0.135 & 0.132 & 0.132 & 0.003 (0.001) & 0.123 & 0.121 & 0.520 & 0.235 \\
\bottomrule
\end{tabular}

\end{table}

In the 33 settings (five designs, five departures, $n\in\{500,1000,2000\}$) in which $\EF$ and $\HL$ differed by more than 0.01, $\EF$ was the most powerful risk-grouped test in all 18 with $A<0$, and the sign of the difference was opposite to that of $A$ in 31 (37 of 39 with the event rate free; Table~S11).

The differences are largest where $|A|$ is large: against the log-log and probit links in D2 ($A=-4.17$ and $-2.15$) $\EF$ gained 0.165 and 0.058, and up to 0.20 at $n=2000$; against the probit link in D4 ($A=-3.80$) $\HL$ fell below its level (0.006) while $\EF$ kept 0.037. Where $A$ is near $+3$, $\HL$ is the more powerful, as the theory says. The alignment tells the analyst which test to use, and Section~\ref{sec:map} gives its sign for asymmetric links.

The Pigeon--Heyse test matched $\HL$ after size adjustment (mean difference 0.001). Tsiatis's test was far more powerful against an omitted interaction (0.90 against 0.08 in D1), which does not change the ordering by risk, and Stukel-W against a wrong link on clean data; Stukel-W misses shared miscalibration (Section~\ref{sec:large-sim}) and several gross errors.

\subsection{Accuracy of the second-order theory}\label{sec:second-check}

The prediction of Theorem~\ref{thm:second} is evaluated in its projection form \eqref{eq:rho} with $k=G-2$, the $\chi^2_{G-2}$ critical value and the vector $\bm m$ of standardized group biases of each departure, from a pseudo-true logistic fit on $10^6$ records. Over 60 settings of the first power study (Fig.~\ref{fig:second}) it predicted a difference larger than 0.01 in 31 and had the right sign in 30; its rank correlation with the simulated gain was 0.97, 0.86 and 0.75 in designs D1, D3 and D4. In the strong-discrimination design it overstates the size of the gain, because groups with few events depart from the Gaussian model. The simulated mean of $C$ had the sign of $A$ in 53 of 54 settings.

\begin{figure}[t]
\centering
\includegraphics[width=84mm]{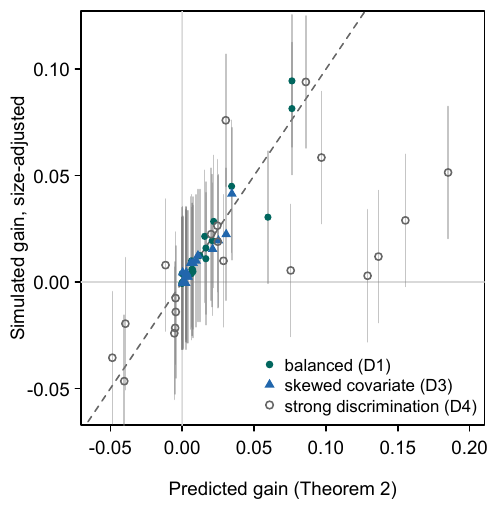}
\caption{Gain in size-adjusted power of $\EF$ over $\HL$: second-order prediction (Theorem~\ref{thm:second}) against simulation in 60 settings (three designs, five departures, $n=200$ to 2000, 2000 replications each). Bars: approximate 95\% Monte Carlo intervals. Dashed line: equality}\label{fig:second}
\end{figure}

\subsection{Where the correction helps}\label{sec:map}

To see where $A<0$, we computed the predicted gain over the Aranda-Ordaz family \citep{ArandaOrdaz1981} $\pi=1-(1+\alpha e^{\eta})^{-1/\alpha}$, the logit at $\alpha=1$ and the complementary log-log as $\alpha\to0$, and over its mirror image $\pi=(1+\alpha e^{-\eta})^{-1/\alpha}$, which tends to the log-log, with $\eta=\mu+sx$, $x\sim N(0,1)$. Fig.~\ref{fig:map} shows a clear pattern: against the complementary log-log side, whose curve rises sharply towards one, $\EF$ gains when most fitted risks are high and loses when they are low; against the log-log side the reverse holds. Averaged over the grid the predicted gain was positive for every $\alpha$.

In the six declared map cells with a predicted difference above 0.02, the simulated gain had the predicted sign, and in the moderate cells also the predicted size (0.043 and 0.058 against 0.057 and 0.058). In the extreme cells the gain was close to the prediction (0.237 against 0.337) while the power of $\HL$ fell to 0.04, below its level (Online Resource~1, Table~S12): there $C$ pulls $\HL$ below its null distribution, the bias of Theorem~\ref{thm:sparse-power}, and $\EF$, which removes $C$, keeps its power.

\subsection{Many small groups}\label{sec:sparse-sim}

With design D1, $n$ up to 16000 and groups of $m\in\{2,5,10,25,50\}$ records formed by fitted risk (Online Resource~1, Tables~S1--S4), the predicted standard deviations of $\EF$ and $\HL$ were within 3\% of the simulated ones whenever there were at least 40 groups, the normal references had size 0.048 to 0.062, and $\chi^2_{G-2}$ gave 0.010 to 0.013 at $m=2$, as Theorem~\ref{thm:sparse-null} predicts.

Fig.~\ref{fig:sparse} shows the power at $n=4000$. Where $A_n<0$ the power of $\EF$ grows with $m$, while the one-sided $\HL$ test is biased for small $m$ (at $n=16000$ and $m=10$, 0.743 against 0.000 for the complementary log-log link); where $A_n>0$, $\HL$ is the more powerful. The first-order predictions follow the simulations closely.

\begin{figure}[t]
\centering
\includegraphics[width=\textwidth]{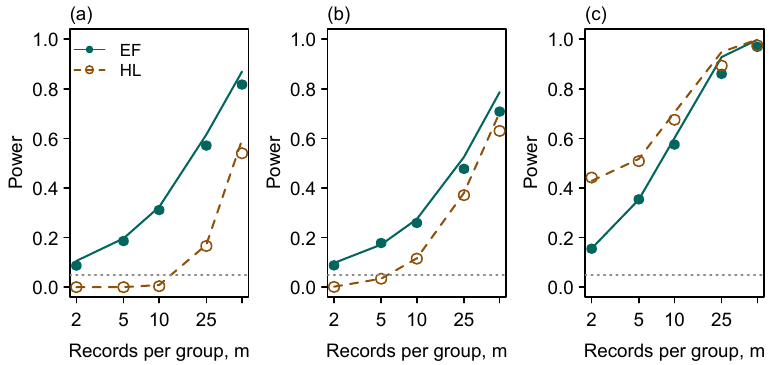}
\caption{Many small groups: power of the one-sided normal tests of Theorem~\ref{thm:sparse-null} at $n=4000$ against the number $m$ of records per group, for (a) the complementary log-log link ($A_n<0$), (b) the log-log link ($A_n<0$) and (c) an omitted square ($A_n>0$). Points: simulation (1000 replications; filled, $\EF$; open, $\HL$); lines: first-order prediction of Theorem~\ref{thm:sparse-power} (solid, $\EF$; dashed, $\HL$). Dotted line: 0.05}\label{fig:sparse}
\end{figure}

\subsection{Large samples: shared miscalibration and the rule of Paul et al.}\label{sec:large-sim}

Three further studies used the normal references with the number of groups of \citet{Paul2013} and with groups of 25 records (Online Resource~1, Tables~S14 and S15).

\emph{Shared miscalibration.} Records were binned by true risk into bins of 25, and each bin received its own deviation $u_b\sim N(0,\tau^2)$ on the logit scale, so that records with similar risk share a miscalibration whose sign varies along the risk scale ($n=4000$). In design D1 the power of $\EF$ was 0.42 and 0.99 at $\tau=0.3$ and 0.6, against 0.11 and 0.28 for Stukel's test and 0.12 and 0.29 for Stukel-W. $\EF$ was at least as powerful as $\HL$ in every setting, as Proposition~\ref{prop:score} predicts, and many groups doubled the power of ten (0.42 against 0.20).

\emph{The rule of Paul et al.} Table~\ref{tab:large} gives the size under the null with their number of groups, which reaches 3202 at $n=20000$. The $\chi^2_{G-2}$ reference, which they recommend, fails in the way Theorem~\ref{thm:sparse-null} predicts from the design constant $R$: it is conservative in design D1 ($R=0.05$; size 0.041 at $n=20000$, predicted 0.037) and liberal for $\HL$ in the strong-discrimination design D4 ($R=11.5$), with size 0.079 at $n=10000$ and 0.117 at $n=20000$, against predicted 0.081 and 0.108. The normal references held the level throughout (0.049 to 0.062). Fig.~\ref{fig:paul} shows the consequence for power. Paul's rule lets the number of groups grow so fast that, where $A<0$, the bias of Theorem~\ref{thm:sparse-power} takes over: against the complementary log-log link in D1 the power of $\HL$ with $\chi^2_{G-2}$ fell from 0.45 at $n=2000$ to 0.06 at $n=20000$, and against the log-log link in D2 it was nil, while $\EF$ with its normal reference reached 0.48 at $n=20000$ in both. $\EF$ was the more powerful in all 16 settings with $A<0$; where $A>0$, $\HL$ kept its advantage.

\emph{Gross errors with many groups.} With five reversed covariates among 5000 records and Paul's 202 groups, $\EF$ with the normal reference had a false-alarm rate of 0.070, against 0.777 for Stukel's test.

\begin{table}[t]
\caption{Size at the 5\% level under a correct model with the number of groups $G$ of \citet{Paul2013} (2000 replications; Monte Carlo standard error 0.005): $\HL$ and $\EF$ referred to $\chi^2_{G-2}$ and to the normal references of Theorem~\ref{thm:sparse-null}. $G$: median over replications. Design constant $R$: 0.05 (D1), 2.98 (D2), 11.5 (D4)}\label{tab:large}
\centering\small
\begin{tabular}{@{}lrrrrrr@{}}
\toprule
 & & & \multicolumn{2}{c}{$\chi^2_{G-2}$} & \multicolumn{2}{c}{normal (Theorem 3)} \\
\cmidrule(lr){4-5}\cmidrule(l){6-7}
Design & $n$ & $G$ & HL & EF & HL & EF \\
\midrule
D1 & 2000 & 34 & 0.051 & 0.051 & 0.061 & 0.060 \\
 & 5000 & 202 & 0.045 & 0.046 & 0.053 & 0.054 \\
 & 10000 & 802 & 0.042 & 0.041 & 0.053 & 0.051 \\
 & 20000 & 3202 & 0.041 & 0.039 & 0.054 & 0.053 \\
\addlinespace
D2 & 2000 & 34 & 0.047 & 0.044 & 0.052 & 0.058 \\
 & 5000 & 202 & 0.056 & 0.050 & 0.059 & 0.058 \\
 & 10000 & 528 & 0.059 & 0.047 & 0.061 & 0.056 \\
 & 20000 & 1056 & 0.057 & 0.046 & 0.056 & 0.052 \\
\addlinespace
D4 & 2000 & 34 & 0.056 & 0.043 & 0.057 & 0.053 \\
 & 5000 & 202 & 0.064 & 0.042 & 0.052 & 0.055 \\
 & 10000 & 802 & 0.079 & 0.051 & 0.056 & 0.060 \\
 & 20000 & 3202 & 0.117 & 0.034 & 0.062 & 0.050 \\
\bottomrule
\end{tabular}

\end{table}

\begin{figure}[t]
\centering
\includegraphics[width=\textwidth]{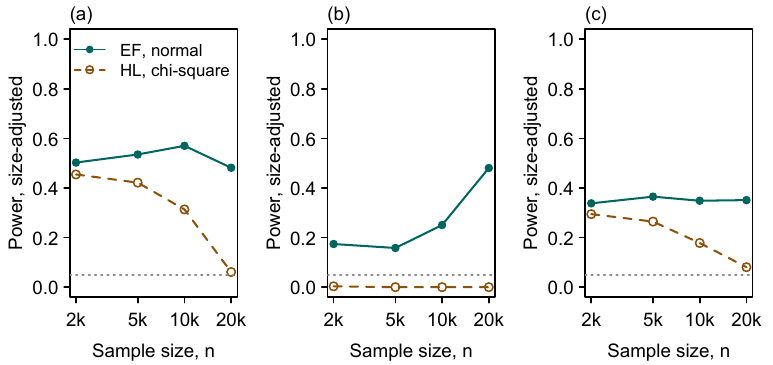}
\caption{Size-adjusted power with the number of groups of \citet{Paul2013} against the sample size, for departures with $A<0$: (a) complementary log-log link in design D1, (b) log-log link in D2, (c) omitted square in D4, each with the event rate of its null (2000 replications). Filled: $\EF$ with the normal reference of Theorem~\ref{thm:sparse-null}; open: $\HL$ with $\chi^2_{G-2}$. Dotted line: 0.05}\label{fig:paul}
\end{figure}

\section{Applications}\label{sec:data}

\subsection{Flour-beetle mortality}\label{sec:beetle}

\citet{Bliss1935} recorded the deaths of 481 flour beetles exposed to eight doses of carbon disulphide, the textbook example of an asymmetric link (data from \pkg{VGAM}, \citealt{Yee2015}); the complementary log-log model fits better than the logistic (AIC 368.7 against 376.5). With one group per dose, the logistic fit gives $\HL=10.01$ ($p=0.124$) and $\EF=10.71$ ($p=0.098$), with $C=-0.70$, the sign of the alignment $A=-0.86$ obtained with the complementary log-log fit as the truth. The read-out (Fig.~\ref{fig:beetle}) locates the misfit at the two highest doses ($c_g=-0.61$ and $-0.98$): the true curve approaches one faster than the logistic. The complementary log-log fit removes the signal ($p=0.77$ for $\HL$, 0.68 for $\EF$).

\begin{figure}[t]
\centering
\includegraphics[width=84mm]{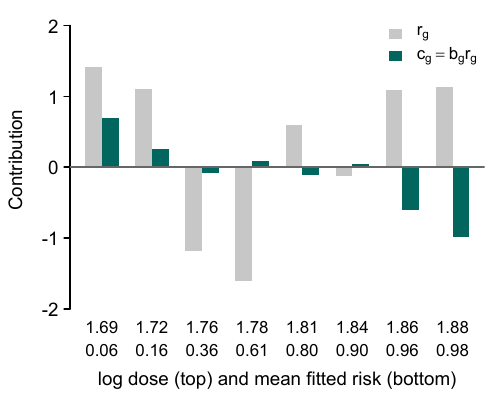}
\caption{Flour-beetle data, logistic fit, one group per dose: standardized group residuals $r_g$ (grey) and directional contributions $c_g=b_gr_g$ (green), whose sum is $C$, by log dose and the mean fitted risk of the dose group}\label{fig:beetle}
\end{figure}

\subsection{Burn injuries: a plasmode study}\label{sec:burn}

We kept the 1000 patients of the burn-injury study \citep{Hosmer2013} with their recorded age, burned area, race, inhalation injury and flame involvement, and simulated the outcome from a logistic, a log-log-type or a complementary log-log-type curve in the risk score of the textbook model, each with the observed death rate of 15\% (2000 replications). Before the run, the map predicted a gain against the log-log-type truth ($A=-6.3$) and a loss against the complementary log-log-type truth ($A=+8.9$).

The simulation confirmed both (Online Resource~1, Table~S18). Against the log-log-type truth $\HL$ was biased (power 0.002 to 0.012) while $\EF$ had 0.19 to 0.31; against the complementary log-log-type truth $\EF$ lost 0.02 to 0.07. Many risks are near zero here, so the reference of $\HL$ was liberal (0.08 to 0.10) and that of $\EF$ conservative, as the rule of Section~\ref{sec:size} anticipates. With up to ten patients' burned area multiplied by four and capped at 100\%, $\EF$ stayed at or below 0.035, while Stukel's test rose to 0.15.

\subsection{A large cohort: the SUPPORT study}\label{sec:support}

The SUPPORT study \citep{Knaus1995} followed seriously ill hospitalized adults. We modelled in-hospital death (2328 deaths among 8873 complete cases) by a logistic model linear in ten covariates, as in earlier work of the authors, and declared the analysis before running it. With ten groups both tests reject the model ($p=4\times10^{-6}$ for $\HL$ and $\EF$). Under the rule of \citet{Paul2013} the data receive 631 groups of 14 records, and both tests accept it ($p=0.78$ and 0.80). Theorem~\ref{thm:sparse-power} explains the reversal (Online Resource~1, Table~S16): measured against a model with a spline in mean blood pressure, the misfit has a pair drift $\Delta_n$ near one at every number of groups, while the null standard deviation grows like $(2G)^{1/2}$. Groups of about 25 records still detect it ($p=0.02$ with 355 groups). Here $R_n=0.018$, so the $\chi^2_{G-2}$ and normal references nearly agree; in a parametric bootstrap of the fitted model their sizes were 0.035 and 0.042.

In seven smaller public data sets (Online Resource~1, Table~S6), neither $\EF$ nor $\HL$ rejected the logistic model at the 5\% level, and $\EF$ had the smaller $p$-value in five of the six besides the beetle data.

\section{Discussion}\label{sec:discussion}

Applied to the risk groups, Farrington's correction turns the Hosmer--Lemeshow statistic into the optimal detector of a shared miscalibration among records of similar risk (Proposition~\ref{prop:score}), against which it was nearly four times as powerful as Stukel's directed test, and it keeps the protection of the risk-ordered table against gross covariate errors. Its theory covers both ten groups and the many groups of large samples, and in the latter it corrects the reference in current use: under the rule of \citet{Paul2013}, the $\chi^2_{G-2}$ reference lets the Hosmer--Lemeshow test reach a size of 0.117 or lose all its power, while the normal references of Theorem~\ref{thm:sparse-null} hold the level and keep the power. The alignment $A$ tells the analyst which of the two risk-grouped tests to use.

The results suggest the following use.
\begin{itemize}
\item When gross recording errors cannot be excluded, prefer a test that groups by fitted risk. Use the grouped Farrington test when the suspected departure has $A<0$ (Fig.~\ref{fig:map}) or is a shared miscalibration without a fixed direction, and the Hosmer--Lemeshow test when $A>0$; with ten groups, each group should expect about one event and one non-event.
\item With many groups, as under the rule of \citet{Paul2013}, refer both statistics to the normal references of Theorem~\ref{thm:sparse-null}, not to $\chi^2_{G-2}$, and prefer groups of about 25 records: they keep the protection against gross errors and avoid diluting a smooth misfit over very many small groups (Section~\ref{sec:support}).
\item Refer the Pigeon--Heyse statistic, if used, to $\chi^2_{G-2}$; with $\chi^2_{G-1}$ it is conservative and otherwise equivalent to $\HL$.
\item When the data have been checked for gross errors and a wrong link is the main concern, Stukel's test or its bounded-leverage version Stukel-W is the more powerful choice, and Tsiatis's test against an omitted interaction.
\end{itemize}

The second-order statements of the fixed-$G$ theory are confirmed by simulation rather than by explicit error bounds, and the many-groups theory is proved for groups that do not depend on the response, with grouping by fitted risk behaving the same in simulation. We did not include the smoothed-residual test of \citet{leCessie1991} or the adaptive test BAGofT \citep{Zhang2023}, because their cost does not suit a routine check: on one core, at $n=1000$, the le Cessie--van Houwelingen test in its original form took 3.8 seconds and BAGofT 8.6 minutes, against 0.004 seconds for $\EF$ (Online Resource~1, Table~S8).

\section*{Computational details}\label{sec:comp}

All computations used R 4.4.1 \citep{RCore2024}, with \pkg{CompQuadForm} \citep{Duchesne2010} for the Imhof integrals, \pkg{ebrahim.gof} 2.9.0 \citep{EbrahimGof} for the tests, \pkg{BAGofT} for the timing of BAGofT, and the data sets of \pkg{aplore3} (burn injuries) and \pkg{VGAM} (flour beetles); the SUPPORT data are the public file \texttt{support2} of the Vanderbilt Department of Biostatistics. An AI assistant (Claude, Anthropic) was used under the authors' direction to write and run simulation and analysis code, to check derivations and to draft parts of the text. The authors designed the study, checked every result against the deposited code and output, and take full responsibility for the content.

\begin{appendices}

\section{Proofs}\label{app:proofs}

\paragraph{Proof of Lemma~\ref{lem:degenerate}.}
With $m_j=1$, a record with $y=1$ contributes $(1-\hat\pi)^2/\{\hat\pi(1-\hat\pi)\}-(1-2\hat\pi)(1-\hat\pi)/\{\hat\pi(1-\hat\pi)\}=\{(1-\hat\pi)-(1-2\hat\pi)\}/\hat\pi=1$, and a record with $y=0$ contributes $\hat\pi^2/\{\hat\pi(1-\hat\pi)\}+(1-2\hat\pi)\hat\pi/\{\hat\pi(1-\hat\pi)\}=\{\hat\pi+(1-2\hat\pi)\}/(1-\hat\pi)=1$. Summing over the $n$ records gives $X^2_\star=n$. \qed

\paragraph{Proof of Lemma~\ref{lem:decomp}.}
Since $e_g=n_g\pib_g$, $o_g-e_g=\sum_{i\in g}\tilde\varepsilon_i$ and $(o_g-e_g)^2=\sum_{i\in g}\tilde\varepsilon_i^2+\sum_{i\neq j\in g}\tilde\varepsilon_i\tilde\varepsilon_j$. For $y_i\in\{0,1\}$, $\tilde\varepsilon_i^2=y_i(1-2\pib_g)+\pib_g^2$, so $\sum_{i\in g}\tilde\varepsilon_i^2=(1-2\pib_g)o_g+n_g\pib_g^2=V_g+(1-2\pib_g)(o_g-e_g)$. Dividing by $V_g$ and summing over $g$ gives $\HL=G+C+Q$, and $\EF=\HL-C=G+Q$. \qed

\paragraph{Proof of Proposition~\ref{prop:score}.}
In the model of \citet{KupperHaseman1978} the probability of $o_g$ events among the $n_g$ records of group $g$ is the binomial probability times $1+\theta\{(o_g-n_g\pib_g)^2+o_g(2\pib_g-1)-n_g\pib_g^2\}/\{2\pib_g^2(1-\pib_g)^2\}$, with $\rho=\theta/\{\pib_g(1-\pib_g)\}$. By the proof of Lemma~\ref{lem:decomp}, $(o_g-n_g\pib_g)^2-o_g(1-2\pib_g)-n_g\pib_g^2=(o_g-e_g)^2-\sum_{i\in g}\tilde\varepsilon_i^2=\sum_{i\neq j\in g}\tilde\varepsilon_i\tilde\varepsilon_j$. The derivative of the log likelihood with respect to $\rho$ at $\rho=0$ is therefore $\sum_g\sum_{i\neq j\in g}\tilde\varepsilon_i\tilde\varepsilon_j/\{2\pib_g(1-\pib_g)\}=S$, and since $V_g=n_g\pib_g(1-\pib_g)$, Lemma~\ref{lem:decomp} gives $\EF-G=\sum_g\sum_{i\neq j\in g}\tilde\varepsilon_i\tilde\varepsilon_j/V_g=\sum_gS_g/n_g$. Under $\rho=0$ the $z_i$ are independent with mean zero, so $\mathrm{E}\{z_iz_j\,\bm x_k(y_k-\pib_{g(k)})\}=\bm0$ for $i\neq j$, because at least one of $i$ and $j$ differs from $k$. The score for $\rho$ is thus uncorrelated with the score for $\bm\beta$, the efficient score for $\rho$ equals $S$, and the $C(\alpha)$ test based on it is asymptotically locally most powerful \citep{CoxHinkley}. \qed

\paragraph{Other proofs.} The proofs of Propositions~\ref{prop:ph}--\ref{prop:influence}, Theorems~\ref{thm:null} and \ref{thm:second} and Corollary~\ref{cor:sign} are in Online Resource~1, Section~S2, and those of Theorems~\ref{thm:sparse-null} and \ref{thm:sparse-power} in Section~S1.

\end{appendices}

\backmatter

\section*{Statements and Declarations}

\noindent\textbf{Funding.} No funding was received for conducting this study.

\smallskip\noindent\textbf{Competing interests.} The authors have no relevant financial or non-financial interests to disclose.

\smallskip\noindent\textbf{Ethics approval.} Not applicable: the study uses simulated data and published data sets.

\smallskip\noindent\textbf{Consent.} Not applicable.

\smallskip\noindent\textbf{Data availability.} The flour-beetle data are in the R package \pkg{VGAM}, the burn-injury data in the R package \pkg{aplore3}, and the SUPPORT data are public (Vanderbilt Department of Biostatistics, file \texttt{support2}); the sources of the other public data sets are listed in Online Resource~1. The simulation code, the declarations with their hashes, every per-replication result and the scripts that produce every table and figure are archived at Zenodo \citep{EbrahimDeposit}.

\smallskip\noindent\textbf{Code availability.} The tests are implemented in the R package \pkg{ebrahim.gof} on CRAN \citep{EbrahimGof}.

\smallskip\noindent\textbf{Author contributions.} E.K.E. proposed the grouped correction, derived the theory, designed and ran the simulations, wrote the software and drafted the manuscript. I.G.K. independently audited the derivations, the simulation code and the reported results. A.E.-K. supervised the work and reviewed and edited the manuscript. All authors read and approved the final manuscript.

\bibliography{references}

\end{document}